\documentclass[aps,pra,reprint,nofootinbib]{revtex4-2}
\usepackage{blindtext}
\usepackage{xcolor}
\usepackage{amsmath, amssymb, amsthm, stmaryrd, mathrsfs}
\usepackage{subcaption}
\usepackage{graphicx}
\usepackage{braket}
\usepackage{enumitem}
\usepackage{bbm}
\theoremstyle{theorem}
\newtheorem{proposition}{Proposition}

\newcommand{\vett}[1]{\mathbf{#1}}
\providecommand{\bgreek}[1]{\mbox{\boldmath$#1$}}

\begin{document}
\title{Single particle entanglement in the mid-and ultra-relativistic regime}
\author{Matthias Ondra}
\author{Beatrix C. Hiesmayr}
\affiliation{University of Vienna, Faculty of Physics, Boltzmanngasse 5, 1090 Vienna, Austria}

\begin{abstract}
In this work we analyze the amount of entanglement associated with the spin and momentum degrees of freedom of a single massive spin-$\frac{1}{2}$ particle from a relativistic perspective. The effect of a Lorentz boost introduces a Wigner rotation that correlates the spin and momentum degrees of freedom. We show that the  natural basis to discuss the geometrical effects of the boost are the helicity eigenstates in the rest frame. In the mid-relativistic regime (where the Wigner rotation angle is limited by $\delta < \frac{\pi}{2}$) we prove for states with equal helicity that the entanglement with respect to the Wigner rotation angle is monotonically decreasing, however, in the  ultra-relativistic regime ($\delta > \frac{\pi}{2}$) the entanglement is increasing. If the states are prepared as a superposition of unequal helicity eigenstates, the monotonic behaviour is inverted.
This implies that in the ultra-relativistic regime a geometrical setup can be found such that the amount of entanglement exhibits local maxima or minima. This shows a counter-intuitive behaviour of the relative amount of entanglement, an effect due to the internal and external geometrical configuration space, and points towards the difficulties in achieving a Lorentz invariant formulation of entanglement in general.
\end{abstract}

\maketitle

\section{Introduction}
Quantum theory and relativity emerged at the beginning of the 20th century and provided explanations to at that time unsolved puzzles in physics. While the theory of relativity deals with the geometry of four-dimensional spacetime, quantum theory describes properties of matter~\cite{peres2004quantum}.
An essential feature of quantum mechanics which allows to discriminate between the ``quantum world'' and classical physics is entanglement~\cite{palge2012generation}. Generally, most of the literature analyzes the theory of entanglement in a non-relativistic setup. To account for a comprehensive understanding of phenomena related to the theory of entanglement, it is inevitable to shed light on entanglement in a relativistic regime. However, extending the concept of entanglement to a relativistic regime turns out to be non-trivial. The underlying different ontologies of quantum mechanics and relativity reveal the extent of difficulties in combining both theories. In classical theory objects (like vectors or fields) and equations of motion of particles are defined in four dimensional spacetime, where objects are expected to behave covariantly. In contrast to that, in quantum theory wave functions are not defined in spacetime but in Hilbert spaces.

 In order to analyze entanglement in a relativistic regime without the machinery of quantum field theory, Peres et.al~\cite{peres2003quantum} distinguished between \textit{primary} variables, which have relativistic transformation laws that depend only on the Lorentz transformation (e.g. the momentum of a particle), and \textit{secondary} variables which have transformation laws that do not depend solely on the Lorentz transformation but also on the momentum of the particle (e.g. the spin of a particle). Peres et.al~\cite{ peres2003quantum, peres2002quantum} presented a first step towards a relativistic extension of quantum information theory, where the authors studied single particle entanglement of a massive spin-$\frac{1}{2}$ particle. One of the key insights is that the reduced spin density matrix of a single spin-$\frac{1}{2}$ particle is not relativistically invariant.
 This is due to the fact that the reduced density matrix for secondary variables has no covariant transformation law, because the spin of the particle undergoes a Wigner rotation whose direction and  magnitude  depends  on  the  momentum  of  the  particle~\cite{peres2004quantum}. Differently stated, the Wigner rotation~\cite{wigner1939unitary} introduces correlation between the spin and the momentum degrees of freedom of the particle, which in a non-relativistic scenario can be considered without loss of generally as being independent. In strong contrast to massless particles as the photon, where the relativistic feature cannot be relaxed and leads that the spin-$1$ particle behaves effectively in some approximations as a spin-$\frac{1}{2}$ particles (polarisation/helicity). For a detailed analyses the reader is referred to Ref.~\cite{Hiesmayr2019}.  Therefore, the entropy as a measure of entanglement has no invariant meaning in special relativity.

 The Wigner rotation in quantum mechanics is closely related  to the Thomas-Wigner rotation in special relativity. Consider a massive particle moving with velocity $u$, which is viewed by a relativistic observer moving with velocity $v$. If the two velocities are not collinear, the resulting Lorentz transformation involves a pure boost and a  rotation (the Thomas-Wigner rotation), which depends on the amounts of the velocities $u$ and $v$, respectively~\cite{sexl2012relativity}. In the quantum version the Wigner rotation~\cite{wigner1939unitary} leaves the momentum invariant, but rotates the spin.

 Recently, the study of entanglement from a relativistic point of view in single-particle and multi-particle scenarios has gained increasing attention~\cite{palge2012generation,peres2002quantum,dunningham2009entanglement,gingrich2002quantum,friis2010relativistic,ren2010quantum,huber2011lorentz,li2003relativistic, Castro_Ruiz_2012}. The multi-particle case has been analyzed e.g. in Gingrich and Adami~\cite{gingrich2002quantum}, where the authors investigated the entanglement of two spin-$\frac{1}{2}$ particles between two systems. It was shown, that for a maximally entangled spin-$\frac{1}{2}$ system the amount of  entanglement decreases if observed from a moving frame of reference. Hence, a Lorentz boost induces a transfer of entanglement between degrees of freedom. The system of two massive spin-$\frac{1}{2}$ particles has been analyzed systematically for all different partitions of the subsystems in Ref.~\cite{friis2010relativistic}. Multipartite entanglement has been studied in Ref.~\cite{huber2011lorentz}, where the authors derived general conditions that have to be met for any Lorentz invariant classification of multipartite entanglement. Various applications of relativistic quantum information have been introduced, ranging from clock synchronization in Ref.~\cite{caban1999lorentz} to relativistic cryptography in Ref~.\cite{czachor2003relativistic}.

In this work we focus on a more theoretical effect of single particle entanglement in a relativistic regime. Dunningham et.~al~\cite{dunningham2009entanglement} considers single particle (modes) entanglement of a massive particle moving in two opposite directions with equal amplitudes and equal speed in a specific geometrical configuration (speeds of the boosts are chosen to be perpendicular) and an initial state which is initially not entangled in the momentum and spin partition. The authors find that the state becomes entangled in the boosted frame of reference and becomes disentangled when the speeds approach the speed of light. Their work has been generalized in Ref.~\cite{palge2012generation} in two ways: (i) by analyzing different scenarios of the boosting geometry and (ii) by assuming Gaussian states with finite width. By analyzing the corresponding scenarios, the authors found that maximally entangled states can be found in a sub-luminar regime depending on the scenario of the boosting geometry.

While the majority of papers discussing effects of relativistic entanglement employ the spin-$\frac{1}{2}$ representation of the Lorentz and Poincare groups, Ref.~\cite{Bittencourt_2019, bittencourt2020single} considered spin-$\frac{1}{2}$ particles as solutions of the Dirac equation, i.e. states are represented by Dirac bispinors. Comparing those two distinct frameworks to describe spin-$\frac{1}{2}$ particles the authors showed that the amount of entanglement in both approaches becomes identical for states with small mass.

In 1927 Heisenberg proposed his famous uncertainty relation demonstrating a universal lower bound to the uncertainties of non-commuting observables. In Ref.~\cite{deutsch1983uncertainty} his inequality was reformulated as a lower bound on the entropies of the observables (improved by~\cite{kraus1987complementary} and~\cite{maassen1988generalized}) and found applications e.g. for security analysis of quantum cryptographic protocols~\cite{coles2017entropic}.  Recently, high energetic and relativistic effects onto those entropic versions have been discussed, e.g. for neutral meson systems exhibiting the effect of $\mathcal{CP}$-violation in Ref.~\cite{di2012heisenberg} and for massive relativistic spin-$\frac{1}{2}$~\cite{bialynicki2019heisenberg}, for photons~\cite{bialynicki2012uncertainty, bialynicki2012heisenberg} and for relativistic bosons~\cite{bialynicki2021heisenberg} which are essentially different to those for fermions.

In this work we consider single massive spin-$\frac{1}{2}$ particles and analyze the entanglement associated with the spin--momentum partition for a boosted observer. We define two classes of states which are distinguished with respect to the helicity in the rest frame, namely (i) states that are superpositions of basis states with equal helicity and (ii) states that are superpositions of basis states with different helicity.  More specifically, we discuss how a change in the inertial frame affects the amount of entanglement in the boosted frame with respect to equal and unequal helicity eigenstates. Those eigenstates allow a discussion of all boosting geometries in a concise and illustrative way and, herewith, we exhaustively discuss all possible effects by a boost with respect to the change in entanglement. In particular, we find that there are two different regimes, the mid- and ultra relativistic regime, given rise to local maxima/minima in the relative entanglement.

The paper is structured as follows. Section~\ref{sec:TW_rotation} discusses the Thomas-Wigner rotation angle from the point of view of special relativity. Section~\ref{sec:Ent} analyzes the entanglement in the boosted frame of reference and reports the results. Section~\ref{sec:Res}  concludes the paper.

\section{Thomas-Wigner rotation in special relativity}\label{sec:TW_rotation}
In a physical world where one is not interested in effects related to gravitation the preferred framework is that of special relativity. One particular effect in special relativity is the Thomas-Wigner rotation which reflects the fact, that only Lorentz boosts in the same direction form a subgroup of the homogeneous Lorentz group. Two consecutive boosts with speeds $u$ and $v$ that are not collinear but enclosing the boosting angle $\phi \in [0, \pi]$ can be decomposed into a pure boost and a pure rotation.
Different approaches to derive this Thomas-Wigner rotation angle $\delta$, ranging from elementary calculations based on the composition of Lorentz boosts up to more involved methodologies based on spinors are presented in Ref.~\cite{o2011elementary}. Perhaps the most well-known formula to compute the Wigner rotation angle is given by~\cite{sexl2012relativity}
\begin{align}
\cos \delta +1 = \frac{(1+ \gamma_u + \gamma_v + \gamma_u\gamma_v(1+ uv \cos \phi))^2}{(\gamma_u +1)(\gamma_v +1)(\gamma_u\gamma_v(1+ uv \cos \phi) +1)}\;. \label{Eq:Wig1}
\end{align}
This equation implies in the relativistic limit of both the particle's and the observer's speed approaching the speed of light that $\delta \to \phi$, which highlights the crucial impact of the boost geometry on the Thomas-Wigner rotation (visualized in Fig.\ref{fig:Wigner_limit}(a)).
\begin{figure*}
    \centering
    \begin{subfigure}[c]{0.3\textwidth}
	\includegraphics[scale=0.65]{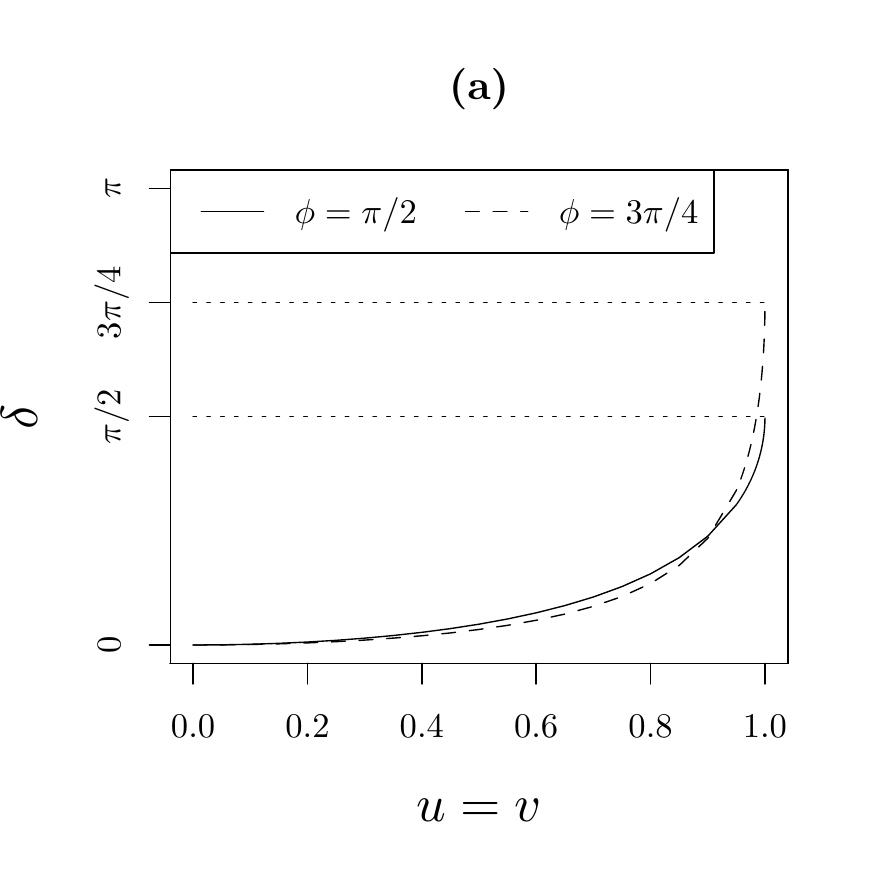}
    \end{subfigure} \hfill
    \begin{subfigure}[c]{0.3\textwidth}
	\includegraphics[scale=0.65]{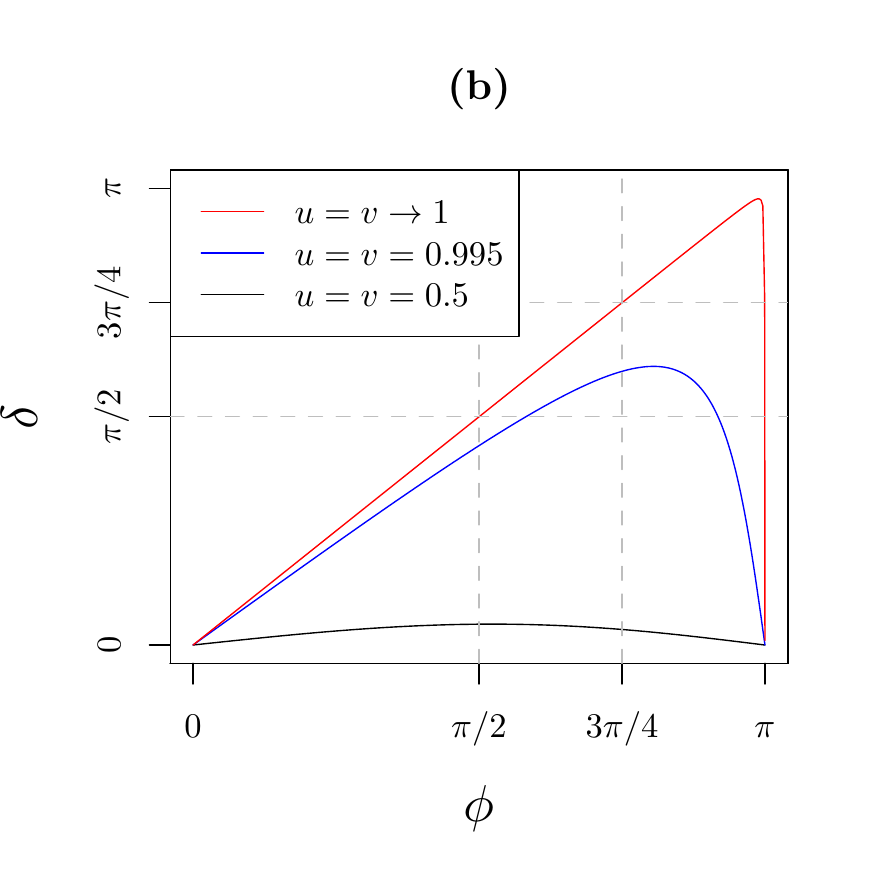}
    \end{subfigure}\hfill
    \begin{subfigure}[c]{0.3\textwidth}
	\includegraphics[scale=0.65]{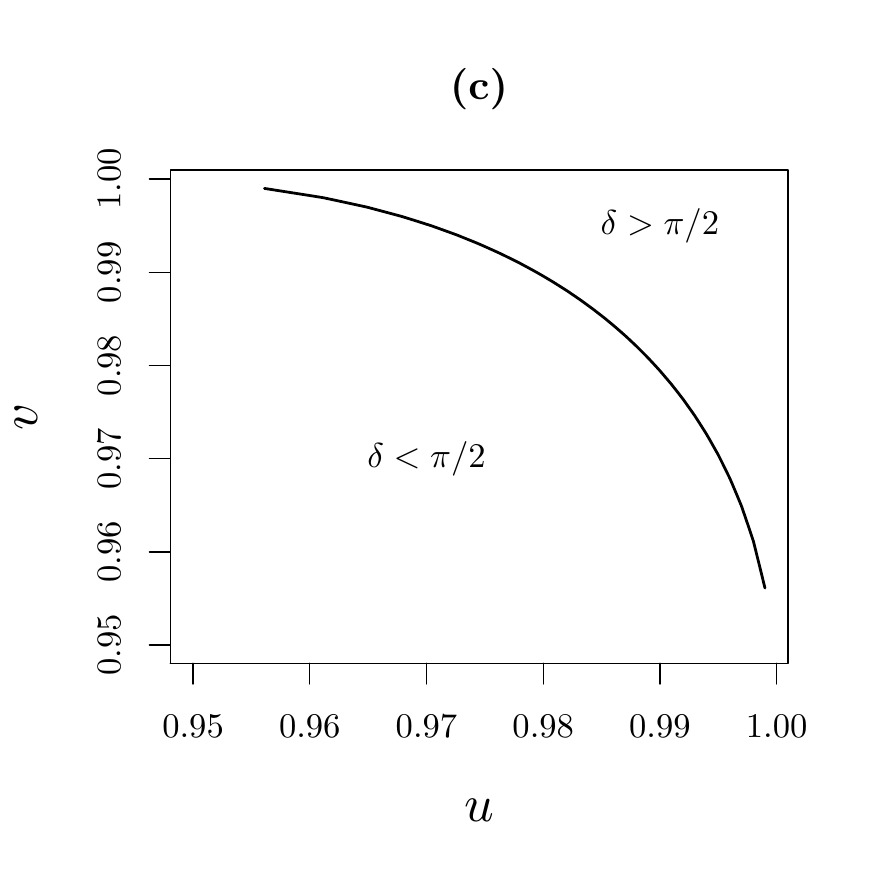}
    \end{subfigure}
    \caption{(a) shows the Wigner rotation angle as a function of the particle's and the observer's speed and different parameter values of the boosting angle for equal speeds $u=v$. Fig.(b) shows the Wigner rotation as a function of the boosting angle for different values of the speeds, where $u=v$. Fig.(c) demonstrates the parameter range of the speeds that refer to Wigner rotation angles $\delta > \pi/2$, where $\phi = 3\pi/4$ is chosen.}
    \label{fig:Wigner_limit}
\end{figure*}
In Ref.~\cite{rhodes2004relativistic} the authors presented a convenient (equivalent) formula of the Wigner rotation angle
\begin{align} \label{Eq:Wig2}
\tan \frac{\delta}{2} = \frac{\sin \phi}{\cos \phi + \mathcal{D}},
\end{align}
where
\begin{align} \label{eq:D}
\mathcal{D}= \sqrt{\frac{(\gamma_u +1)(\gamma_v +1)}{(\gamma_u -1)(\gamma_v -1)}} \geq 1
\end{align}
captures the full dependence on the speeds $u$ and $v$, respectively. Let us now analyze the Wigner rotation angle as a function of the boosting angle, for fixed speeds $u$ and $v$. An interesting characteristic of the Wigner rotation angle is, that it admits a unique maximum at the boosting angle
\begin{align} \label{eq:argmax}
\phi^* &= \arccos\left( -\frac{1}{\mathcal{D}} \right).
\end{align}
This can be seen by isolating $\delta$ in \eqref{Eq:Wig2} and setting the first derivative to zero. Uniqueness follows from the fact, that $\delta(\phi)$ is globally concave. The maximum effect of the Wigner rotation occurs at a boosting angle $\phi \geq \pi/2$, see Fig.\ref{fig:Wigner_limit}(b). The maximum effect of the Wigner rotation does not occur when the boosts are chosen perpendicular, but when the boosting directions are chosen slightly anti-aligned.
Furthermore, we note that for the special boosting geometry $\phi = \pi/2$, the Wigner rotation angle $\delta = \pi/2$ can only achieved in the theoretical limit of both speeds approaching the speed of light. In contrast to that, for a boosting angle $\phi \geq \pi/2$ values of the Wigner rotation angle $\delta \geq \pi/2$ can be achieved in a sub-luminar regime. From \eqref{Eq:Wig2} we find that parameter values $u,v$ and $\phi$ satisfying
\begin{align}
    1- \sin(2 \phi) \geq \mathcal{D}^2
\end{align}
refer to the Wigner rotation angle $\delta \geq \pi/2$, see Fig.\ref{fig:Wigner_limit}(c). However, as it is illustrated in Fig.\ref{fig:Wigner_limit} Wigner rotation angles $\delta > \pi/2$ are observed in the ultra-relativistic regime, whereas $\delta< \pi/2$ can be observed in the mid-relativistic regime.

\section{Entanglement of a massive particle}\label{sec:Ent}
In this section we discuss first the state space of a single massive particle and the relation between the outer degrees of freedom (momentum) and the inner degrees of freedom (spin). Then we proceed in defining the entanglement in the rest frame and in the boosted frame.

\subsection{Setting the stage}
The quantum mechanical description of a quantum particle with spin $s$ propagating in the real $3$-dimensional space $\mathbb{R}^{3}$ with momentum $\vett{p}$ takes place in the Hilbert space
\begin{eqnarray}
\label{theextendedspace}
\mathcal{H}
 = &&\bigg\{ \bgreek{\psi}(\vett{p}) : \bgreek{\psi}(\vett{p}) = \sum_{i=1}^3 f^i(\vett{p}) \otimes \vett{e}_i \, \biggl{|} \nonumber\\ &&\vett{e}_i\in\mathbb{C}^{2s+1}\;\textrm{and}\;\,
f^i(\vett{p}) \in \mathcal{L}^2\left( \mathbb{R}^3,\frac{d^3\vett{p}}{|\vett{p}|}\right)
\bigg\},
\end{eqnarray}
which is isomorphic to $\simeq \mathcal{L}^2\left( \mathbb{R}^3,\frac{d^3\vett{p}}{|\vett{p}|}\right) \otimes \mathbb{C}^{2s+1}$. The very fact that one can or has to distinguish between the outer degrees of freedom (momentum) and the inner degrees of freedom (spin) implies a tensor product structure of the underlying Hilbert space and therefore the states can be factorized, i.e. a notion of separability/entanglement is straightforwardly superimposed. This substructure is the one that we analyse from a relativistic observer point of view in the following sections.

However, in general the internal spin space cannot be considered totally independent of the propagation direction of the particles, which e.g. manifests in the famous spin-orbital couplings. In mathematical terms this means that (i) the spatial/momentum wave functions have to be related to the momenta $\vett{e}_i\longrightarrow \vett{e}_i(\vett{p})$ and (ii) a connection of rotations in $\mathbb{R}^3$ and rotations along the three spatial directions $i$ is associated to ${\vett{e}}_i(\vett{p})$, namely
\begin{eqnarray}\label{rotation}
{\mathrm R}& =& \hat{V}^{\dagger}\, e^{-\frac{i}{\hbar}\, \varphi\, \vett{n} \cdot \vett{\hat{S}}}\, \hat{V},
\end{eqnarray}
where $\mathrm{R}\in SO(3)$ are (real) rotation matrices describing a rotation in the real $3$-dimensional space $\mathbb{R}^3$, $\vett{\hat{S}}$ are the generators of $SO(3)$ vector rotations a subgroup of SU(N) with $N=2s+1$, $\hat{V}$ are proper unitary matrices chosen such that $\mathrm{R}$ is real and $\vett{n}$ describes the rotation axis chosen  in $\mathbb{R}^3$ and $\varphi$ the angle. For spin-$\frac{1}{2}$ there exists an one-to-one correspondence between the rotations  $\mathrm{R}\in SO(3)$ and the rotations in the spin space generated by elements of $SU(2)$, the three famous Pauli matrices. For higher spins there are more generators of $SU(N)$, i.e. $(2s+1)^2-1$. For instance, in the case of $s=1$-particles there are eight hermitian matrices, the Gell-Mann matrices, here e.g. the three antisymmetric Gell-Mann matrices with $\hat{V}=\mathbbm{1}_3$ are a proper choice for the spin operator $\vett{\hat{S}}$, a $SU(2)$ subgroup of $SU(3)$ with a one-to-one correspondence to the rotations $\mathrm{R} \in SO(3)$.

A convenient choice of the three vectors $\{\vett{e}_i(\vett{p})\}$ is to identify one with the direction of propagation direction, e.g. with the third component, namely $\vett{e}_3(\vett{p})=\frac{\vett{p}}{|\vett{p}|}$ (by that also choosing w.l.o.g. a real vector in $\mathbb{C}^{2s+1}$). Requiring an orthonormal coordinate system this further implies $\vett{e}_1(\vett{p}) \times \vett{e}_2(\vett{p}) = \vett{e}_3(\vett{p})$. Applying this rule to all momenta, the above defined Hilbert space $\mathcal{H}$ can be equipped with a positive-definite scalar product of the form
\begin{eqnarray}
\langle {\bgreek{\phi}} |  {\bgreek{\psi}}\rangle:=
\int \frac{d^3\vett{p}}{|\vett{p}|} \, &&(
g^{1}(\vett{p})^* \, f^1(\vett{p}) +
g^{2}(\vett{p})^* \, f^2(\vett{p}) \nonumber \\
&&+g^{3}(\vett{p})^* \, f^3(\vett{p}))\;,
\end{eqnarray}
since $\vett{e}^{*}_i(\vett{p})\cdot \vett{e}_j(\vett{p})=\delta_{i,j}$ for all $\vett{p}$.

For instance, if considering massless  spin-$1$ quantum particles such as photons further constrains as the Maxwell's equations have to be satisfied which result in a folding of the momentum and spin degrees of freedom leading to the concept of polarisation and helicity. From the quantum information theoretic perspective this can be handled in a concise and consistent unified way and was presented in Ref.~\cite{Hiesmayr2019}.

In this contribution we focus on massive spin-$\frac{1}{2}$ particles, where no per se correlations between the degrees of freedom from momentum and spin account, however, as we show in the following changing the reference frame a correlation between those degrees of freedom for certain initial states are generated, when the reference frame is boosted. For a single massive spin-$\frac{1}{2}$ particle the spin can only take two distinct values  $s \in \{ \uparrow, \downarrow \}$ and the four momentum is given by $p^\mu = (p^0=|\vett{p}|, \mathbf{p})^T$, with $p^\mu p_\mu = (p^0)^2 - \mathbf{p}^2 =m^2$, where $m>0$ denotes the mass of the particle. From now on we switch to the Dirac notations and an arbitrary state is denoted by $|\psi\rangle=\ket{p}\otimes \ket{s}\equiv \ket{p s}$.

We consider sufficiently peaked (i.e. sharp) momentum states and assume that $p^\mu$ can only take two values $p_+$ and $p_-$, which are assumed to be in opposite directions, i.e. along the same axis.
 Sharp momenta are modeled via delta-distributions in the momentum space and result in single Wigner rotations.
This assumption allows us to regard the system of a single spin-$\frac{1}{2}$ particle as a two-qubit system. In particular, the possible momentum states $p_\pm$ of the particle in its rest frame $\Sigma$ are chosen along the $\pm z-$axis, which is also assumed to be the spin-quantization axis. Without loss of generality the boost of the moving observer $\Sigma^\prime$ is in the $x-z$ plane and the boosting direction encloses an angle $\phi$ with the $z-$axis. Therefore, the Wigner rotation results in a rotation with angle $\delta$ around the $y-$axis. This setup is illustrated in Fig.\ref{fig:setup}.
\begin{figure}
    \centering
    \includegraphics[width=0.6\textwidth]{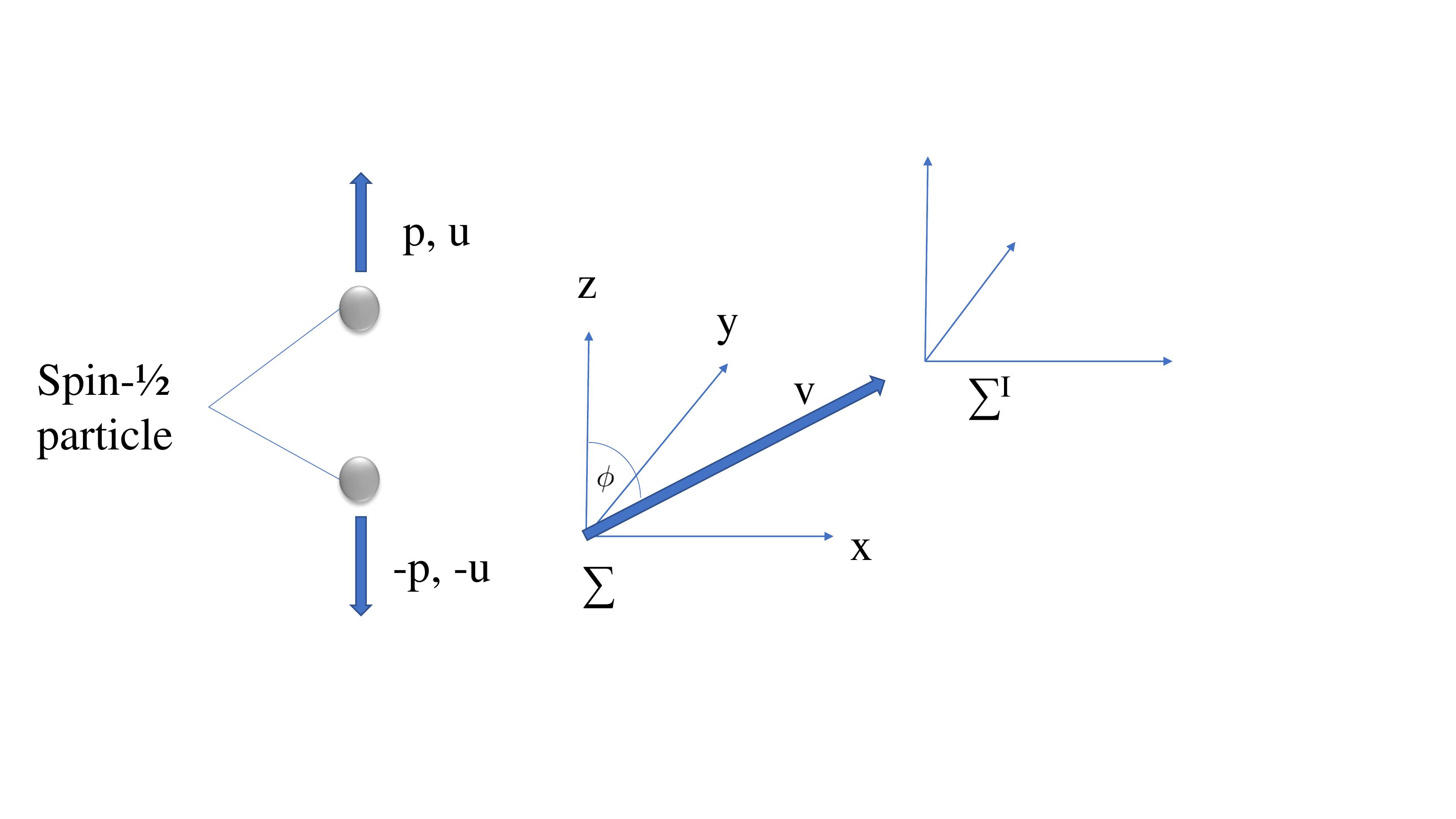}
    \caption{This figure illustrates the setup, where the particle travels along the $z$-axis in the rest frame $\Sigma$ and the moving observer $\Sigma^\prime$ is boosted with velocity $v$ along an axis which encloses the boosting angle $\phi$ with the $z-$axis.}
    \label{fig:setup}
\end{figure}

\subsection{Entanglement in the rest frame $\Sigma$}
In our analysis we investigate states in the rest frame that are linear combinations of equal helicity as well as a state with unequal helicity.
Without loss of generality our initial spin-momentum state in its rest frame $\Sigma$ in case of equal helicity components is defined with $\eta \in [0, 2\pi)$ via
 \begin{align} \label{eq:initialState}
\ket{\psi}_\Sigma = \cos \eta \ket{p_+, \uparrow} + \sin \eta \ket{p_-, \downarrow}\;.
\end{align}
This state represents a combination of states with helicity $h=1$. However, as we demonstrate later the state   $\ket{\tilde{\psi}}_\Sigma = \cos \eta \ket{p_+, \downarrow} + \sin \eta \ket{p_-, \uparrow}$ constituted by a linear combination of states with $h=-1$ gives the same results. In case of a state with unequal helicity we consider the state
\begin{eqnarray}\label{eq:initialStatehm1}
    \ket{\xi}_\Sigma &&=\cos \eta \ket{p_+, \uparrow} + \sin \eta \ket{p_-, \uparrow} \nonumber\\
    &&= (\cos \eta \ket{p_+} + \sin \eta \ket{p_-}) \otimes \ket{\uparrow},
\end{eqnarray}
which is initially not entangled in the rest frame.

In order to quantify the amount of entanglement distributed within the spin momentum partition of the state in the rest frame $\Sigma$, we use the von Neumann entropy $E=-\sum p_i \log_2 p_i$ (any Renyi entropy does the job since we are dealing only with pure states). The entanglement $E^\psi_\Sigma(\eta)$ depends in its rest frame $\Sigma$ only on $\eta$ and derives to
\begin{align}\label{eq:Entanglement_Rest_Frame}
E^{\tilde{\psi}}_\Sigma(\eta) = E^\psi_\Sigma(\eta) = -\cos^2 \eta \log(\cos^2 \eta) - \sin^2 \eta \log(\sin^2 \eta)\;.
\end{align}
Therefore, the state-preparation parameter $\eta$ induces initial entanglement within the spin-momentum partition in the particle's rest frame $\Sigma$ if $\eta \not= n \frac{\pi}{2}$, $n \in \mathbb{N}$. For $\eta = (2n+1) \frac{\pi}{4}$, $n \in \mathbb{N}$ we recover the maximally entangled Bell states
\begin{align}\ket{\Psi^{\pm}} = \frac{1}{\sqrt{2}}(\ket{p_+, \uparrow} \pm  \ket{p_-, \downarrow})\;.\end{align}
For the state with equal helicity \eqref{eq:initialStatehm1} we find $E^\xi_\Sigma =0$, irrespective of the value $\eta$.
\subsection{Entanglement in the boosted frame $\Sigma^\prime$}
In order to investigate how the amount of entanglement changes when observed from a moving observer, we consider setups including a non-vanishing Wigner rotation, i.e. the boosts are not collinear $\phi \neq 0$.  The appropriate relativistic transformation of quantum states is the Wigner rotation which only acts on the spin degree of freedom introduced by Wigner~\cite{wigner1939unitary}.
In case of massive particles with four-momentum $p^\mu = (m, \boldsymbol{0})^T$, the elements of Wigner's little group are the 3D rotations. Therefore, a suitable representation of the Wigner's little group is a unitary matrix $SU(2)$ represented by
\begin{align} \label{eq:repr}
U = \mathbf{1} \cos \frac{\delta}{2} + i (\boldsymbol{\sigma}\cdot \boldsymbol{n}) \sin \frac{\delta}{2},
\end{align}
where $\mathbf{n}$ denotes the axis of rotation, i.e. $\mathbf{n}$ is a 3D unit vector.
Therefore, the spin-momentum state in the boosted frame of reference $\Sigma^\prime$ is given by
\begin{align}
\ket{\chi}_{\Sigma^\prime} = U(\Lambda, p_\pm) \ket{\chi}_\Sigma,
\end{align}
for all states $\chi \in \{\psi, \tilde{\psi}, \xi \}$ considered, where $U(\Lambda, p_\pm)$ is the unitary representation of the Wigner rotation on the state in its rest frame. The matrices performing the rotation of the spin with the rotation axis $\mathbf{n}=(0,1,0)^T$ for the momenta $p_\pm$ are given by
\begin{align}
U_\pm =
\begin{pmatrix}
\cos \frac{\delta}{2} & \pm \sin \frac{\delta}{2} \\
\mp \sin \frac{\delta}{2} & \cos \frac{\delta}{2}
\end{pmatrix}.
\end{align}
\subsubsection{The case of equal helicity}
Let us first consider the state with helicity $h=1$, i.e. the state $\ket{\psi}_\Sigma$. The calculation of the boosted state in the moving frame of reference $\Sigma^\prime$ is
\begin{eqnarray} \label{eq:state_boosted}
\ket{\psi}_{\Sigma^\prime} &&= \cos \eta \cos \frac{\delta}{2} \ket{p_+^\prime, \uparrow} -\cos \eta \sin \frac{\delta}{2} \ket{p_+^\prime, \downarrow} \nonumber \\
&&- \sin \eta	 \sin \frac{\delta}{2} \ket{p_-^\prime, \uparrow} + \sin \eta \cos \frac{\delta}{2} \ket{p_-^\prime, \downarrow},
\end{eqnarray}
where $p^\prime_\pm = \Lambda p_\pm$ denotes the transformed four-momentum of the particle.

The amount of entanglement in the moving frame of reference derives to
\begin{eqnarray}\label{Eq:Wigner_moving}
E^\psi_{\Sigma^\prime} &&(\delta) =-\frac{1}{2} \bigg( 1 + \sqrt{\cos^2 (2\eta) + \sin^2 \delta \sin^2 (2\eta)} \bigg) \nonumber \\
 &&\times\log_2\left( \frac{1}{2} \left( 1 + \sqrt{\cos^2 (2\eta) + \sin^2 \delta \sin^2 (2\eta)} \right)\right) \nonumber\\
&&-  \frac{1}{2} \left( 1 - \sqrt{\cos^2 (2\eta) + \sin^2 \delta \sin^2 (2\eta)} \right) \nonumber \\
 &&\times\log_2\left( \frac{1}{2} \left( 1 - \sqrt{\cos^2 (2\eta) + \sin^2 \delta \sin^2 (2\eta)} \right)\right).
\end{eqnarray}
By continuity of the Wigner rotation, we should recover the entanglement in the rest frame in the limit  $\lim_{\delta \to 0} E^\psi_{\Sigma^\prime}(\delta)= E^\psi_\Sigma$, which is obviously the case.

Let us now consider the state with helicity $h=-1$, i.e. the state $\ket{\tilde{\psi}}$. Similar to the case considered above we find that the state in the moving frame of reference $\Sigma^\prime$ is
\begin{eqnarray}
    \ket{\tilde{\psi}}_{\Sigma^\prime} &&= \cos \eta \sin \frac{\delta}{2} \ket{p_+^\prime, \uparrow} +\cos \eta \cos \frac{\delta}{2} \ket{p_+^\prime, \downarrow} \nonumber \\
    &&+\sin \eta \cos \frac{\delta}{2}  \ket{p_-^\prime, \uparrow} +\sin \eta \sin \frac{\delta}{2} \ket{p_-^\prime, \downarrow}\;.
\end{eqnarray}
Due to the fact, that the state $\ket{\psi}_{\Sigma^\prime}$ can be transformed via a local unitary operation into the state $\ket{\tilde{\psi}}_{\Sigma^\prime}$, i.e. $\ket{\tilde{\psi}}_{\Sigma^\prime} = -i \sigma_Z \otimes \sigma_Y \ket{\psi}_{\Sigma^\prime}$, the amount of entanglement in both cases is the same $E^{\tilde{\psi}}_{\Sigma^\prime}(\delta)= E^\psi_{\Sigma^\prime}(\delta)$. Therefore, both states can be considered as equivalent in terms of their entanglement properties.
\subsubsection{The case of unequal helicity}
Let us now consider the case of $\ket{\xi}_\Sigma$, i.e. the state being in a superposition of unequal helicity states. Although the state is initially not entangled in the rest frame, the Wigner rotation introduces correlation in the spin-momentum partition. We find, that the state in the moving frame of reference $\Sigma^\prime$ is
\begin{eqnarray}
    \ket{\xi}_{\Sigma^\prime} &&= \cos \eta \cos \frac{\delta}{2} \ket{p_+^\prime, \uparrow} -\cos \eta \sin \frac{\delta}{2} \ket{p_+^\prime, \downarrow} \nonumber \\
    && +\sin \eta \cos \frac{\delta}{2}  \ket{p_-^\prime, \uparrow} +\sin \eta \sin \frac{\delta}{2} \ket{p_-^\prime, \downarrow}\;.
\end{eqnarray}
Comparing this state with the state $\ket{\psi}_{\Sigma^\prime}$ \eqref{eq:state_boosted} we find that we can reproduce $\ket{\xi}_{\Sigma^\prime}$ by applying a controlled-$U$ gate with $U= i \sigma_Y$, i.e. basically a $\mathrm{CNOT}$ gate with an additional phase shift. The amount of entanglement in the moving frame of reference derives to
\begin{eqnarray}
E^\psi_{\Sigma^\prime}&&(\delta) = -\frac{1}{2} \left( 1 + \sqrt{\cos^2 (2\eta) + \cos^2 \delta \sin^2 (2\eta)} \right) \nonumber\\
 &&\times\log_2\left( \frac{1}{2} \left( 1 + \sqrt{\cos^2 (2\eta) + \cos^2 \delta \sin^2 (2\eta)} \right)\right)\nonumber \\
&&-  \frac{1}{2} \left( 1 - \sqrt{\cos^2 (2\eta) + \cos^2 \delta \sin^2 (2\eta)} \right) \nonumber \\
 &&\times\log_2\left( \frac{1}{2} \left( 1 - \sqrt{\cos^2 (2\eta) + \cos^2 \delta \sin^2 (2\eta)} \right)\right).
\end{eqnarray}
\section{Results}\label{sec:Res}
Given a state $\ket{\chi}$ with a given entanglement $E^\chi_\Sigma(\eta)$, where $\chi \in \{\psi, \tilde{\psi}, \xi \}$ we observe that the change in the boosted frame $\Sigma'$ does only dependent on the Wigner rotation angle $\delta$. Moreover, we find a monotonic behaviour of the amount of entanglement in the boosted frame with respect to a change in the Wigner rotation angle. Surprisingly, the monotonic behaviour is different when considering states with equal helicity and states with unequal helicity. More specifically, the find that the monotonic behaviour is inverted.
\begin{proposition} \label{prop:1}
For any state $\ket{\psi}_\Sigma$ in the rest frame $\Sigma$ the amount of entanglement $E^\psi_{\Sigma^\prime}$ in the moving frame is:
\begin{enumerate}[label=(\roman*)]
    \item decreasing with the Wigner rotation angle $\mathrm{d} E^\psi_{\Sigma^\prime}/ \mathrm{d} \delta \leq 0$ for $\delta \in (0, \pi/2)$
    \item increasing with the Wigner rotation angle $\mathrm{d} E^\psi_{\Sigma^\prime}/ \mathrm{d} \delta \geq 0$ for $\delta \in (\pi/2, \pi)$.
\end{enumerate}
For any state $\ket{\xi}_\Sigma$ in the rest frame $\Sigma$ the amount of entanglement $E^\xi_{\Sigma^\prime}$ in the moving frame is:
\begin{enumerate}[label=(\roman*)]
    \item increasing with the Wigner rotation angle $\mathrm{d} E^\xi_{\Sigma^\prime}/ \mathrm{d} \delta \geq 0$ for $\delta \in (0, \pi/2)$
    \item decreasing with the Wigner rotation angle $\mathrm{d} E^\xi_{\Sigma^\prime}/ \mathrm{d} \delta \leq 0$ for $\delta \in (\pi/2, \pi)$.
\end{enumerate}
\end{proposition}
\begin{proof}
We prove the results for the state $\ket{\psi}_\Sigma$, the state $\ket{\xi}_\Sigma$ can be treated in a similar way. The amount of entanglement distributed within the spin--momentum partition in the moving frame of reference is given by \eqref{Eq:Wigner_moving}. We define
\begin{align}
p(\delta) = \frac{1}{2} \left( 1 + \sqrt{\cos^2 (2\eta) + \sin^2 \delta \sin^2 (2\eta)} \right)
\end{align}
and rewrite the logarithm  $\log_2{x} = \ln x / \ln 2$. The derivative with respect to the Wigner rotation angle is
\begin{align} \label{eq:Wigner_der}
\begin{split}
\frac{\mathrm{d} E_{\Sigma^\prime}(\delta)}{\mathrm{d} \delta}
&= \frac{1}{\ln 2}\frac{\mathrm{d}}{\mathrm{d} \delta} \left( -p(\delta) \ln(p(\delta)) - (1-p(\delta)) \ln(1-p(\delta)) \right) \\
&= \frac{1}{\ln 2} p^\prime(\delta) \ln \left( \frac{1-p(\delta)}{p(\delta)} \right) \\
&= \frac{1}{\ln 2} \dfrac{\sin(2\delta) \sin^2(2\eta)}{2\sqrt{\cos^2 (2\eta) + \sin^2 \delta \sin^2 (2\eta)}}  \ln \left( \frac{1-p(\delta)}{p(\delta)} \right).
\end{split}
\end{align}
Let us first analyze the case $\delta \in (0, \pi/2)$. To observe that the entanglement is decreasing, we observe that $p > 1/2$. Therefore, $\ln((1-p)/p) < 0$ and we have that $\mathrm{d} E_{\Sigma^\prime}/\mathrm{d}\delta \leq 0$
for $\delta \in (0, \pi/2)$.
The case $\delta \in (\pi/2, \pi)$ is treated in a similar way and we find that $\mathrm{d} E_{\Sigma^\prime}/\mathrm{d}\delta \geq 0$ holds true in this range.
\end{proof}
The inequalities in Proposition \ref{prop:1} are strict, whenever (i) the Wigner rotation angle $\delta \notin \{0, \pi/2, \pi \}$ and (ii) $\eta\not=n\; \frac{\pi}{2}$ (the spin-momentum state $\ket{\psi}_\Sigma$ in the rest frame is not separable or the momentum state of $\ket{\xi}_\Sigma$ not in a superposition of momentum states).

Let us now discuss these parameter values of the Wigner rotation angle for the state $\ket{\psi}_\Sigma$ in more detail. The scenario of a vanishing Wigner rotation angle $\delta=0$ corresponds to the trivial case of either the particle's or the observer's speed vanish. The case $\delta= \pi$ is also not of major interest due to the fact that it can only be obtained in the relativistic limit $u,v \to 1$. The most interesting parameter value is the threshold value of $\delta = \pi/2$, where the  monotonic behaviour of the entanglement in the boosted frame is inverted. For this specific parameter value, the spin--momentum state $\ket{\psi}_\Sigma$ in the boosted frame of reference is given by
\begin{eqnarray}
\ket{\psi}_{\Sigma^\prime} &&= \frac{\cos \eta}{\sqrt{2}} ( \ket{p_+^\prime, \uparrow} -\ket{p_+^\prime, \downarrow})+ \frac{\sin \eta}{\sqrt{2}}  ( -\ket{p_-^\prime, \uparrow} +  \ket{p_-^\prime, \downarrow}) \nonumber\\
&&= (\cos \eta\ket{p_+^\prime} - \sin \eta \ket{p_-^\prime}) \otimes \frac{1}{\sqrt{2}}(\ket{\uparrow}- \ket{\downarrow}).
\end{eqnarray}
Therefore, at the threshold value $\delta= \pi/2$ the state in the boosted frame of reference becomes disentangled irrespective of the state preparation in the rest frame, i.e. if the state was entangled or not.

In case of the state $\ket{\xi}_\Sigma$ we find similar characteristics. For $\delta=0$ and $\delta= \pi/2$ the entanglement is the same as in the rest frame. For $\delta=\pi/2$ we find that
\begin{align}
    \ket{\xi}_{\Sigma^\prime} &&= \frac{\cos \eta}{\sqrt{2}}  (\ket{p_+^\prime, \uparrow} - \ket{p_+^\prime, \downarrow})+ \frac{\sin \eta}{\sqrt{2}}  ( \ket{p_-^\prime, \uparrow}+ \ket{p_+^\prime, \downarrow}),
\end{align}
which has the reduced density matrix
\begin{align}
 \rho =   \begin{pmatrix}
    \cos^2 \eta & 0 \\
    0 & \sin^2 \eta
    \end{pmatrix} \;.
\end{align}
Hence, whenever $\eta= \pi/4$ we recover the maximally mixed state. Whereas in the case of equal helicity a Wigner rotation angle $\delta = \pi/2$ disentangles the state, in the case of unequal helicity we can observe maximum entanglement at $\delta= \pi/2$.

As it is demonstrated in Fig.\ref{fig:Wigner_limit}, Wigner rotation angles $\delta \geq \pi/2$ can only be achieved in the ultra-relativistic regime for an appropriate boosting angle $\phi \geq \pi/2$, i.e. the boosting angle has to be chosen large enough. However, for a wide range of parameter values in the mid-relativistic regime the threshold value of the Wigner rotation angle cannot be attained. This is of special interest, when the entanglement in the boosted frame as a function of the boosting angle becomes $\tilde{E}^\chi_{\Sigma^\prime}(\phi) = E^\chi_\Sigma(\delta(\phi))$ and the mid-relativistic regime is considered. The monotonic behaviour of the Wigner rotation angle as a function of the boosting angle carries over to the properties of the entanglement in the boosted frame as a function of the boosting angle. Hence, in case of $\ket{\psi}_\Sigma$ the minimum of the amount of entanglement is observed whenever the effect of the Wigner rotation is maximal, whereas for $\ket{\xi}_\Sigma$ the maximum of the amount of entanglement is observed whenever the effect of the Wigner rotation is maximal.
\begin{figure*}
\begin{center}
    \begin{subfigure}[c]{0.3\textwidth}
    \includegraphics[scale=0.65]{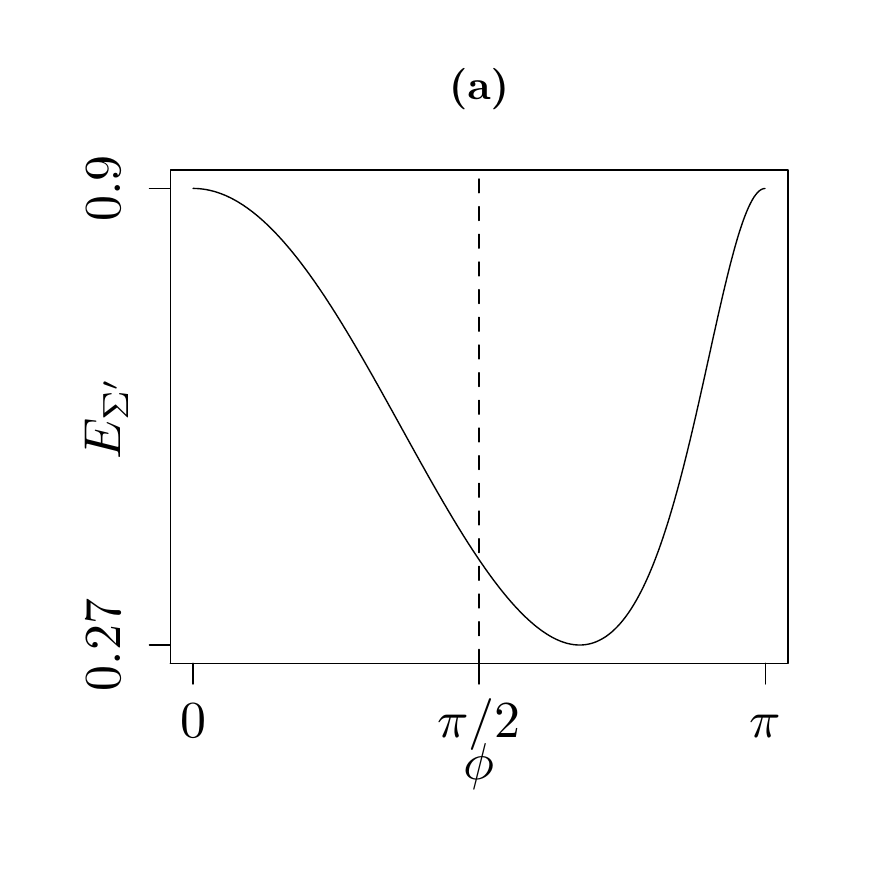}
    \end{subfigure}\hfill
    \begin{subfigure}[c]{0.3\textwidth}
	\includegraphics[scale=0.65]{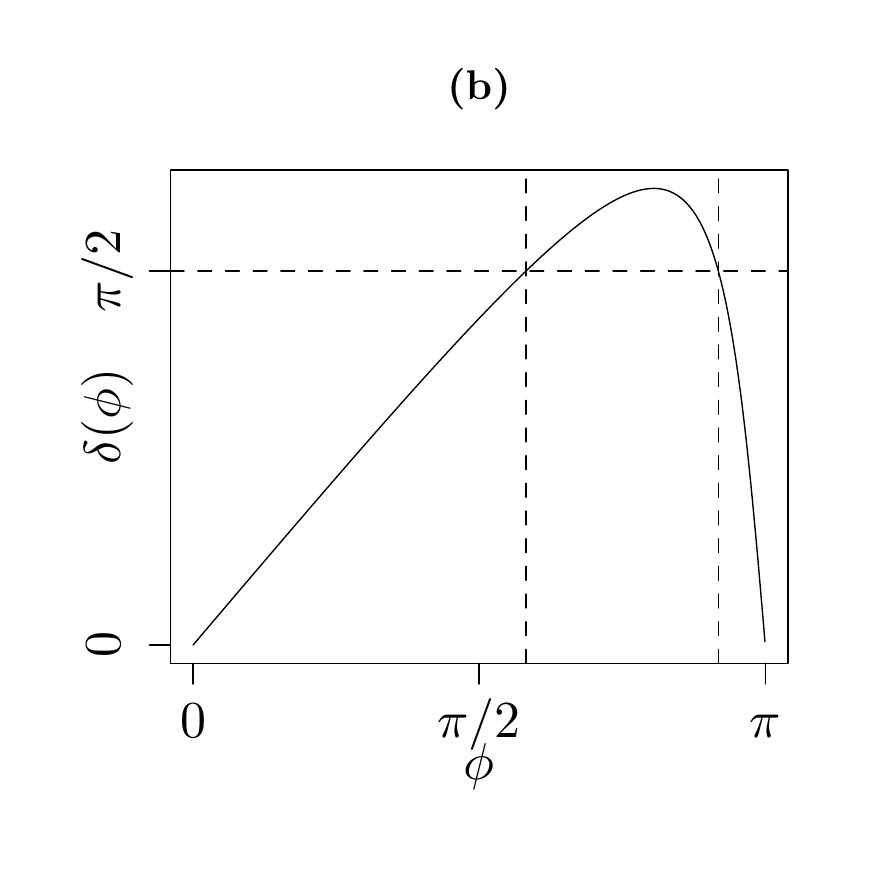}
    \end{subfigure} \hfill
    \begin{subfigure}[c]{0.3\textwidth}
	\includegraphics[scale=0.65]{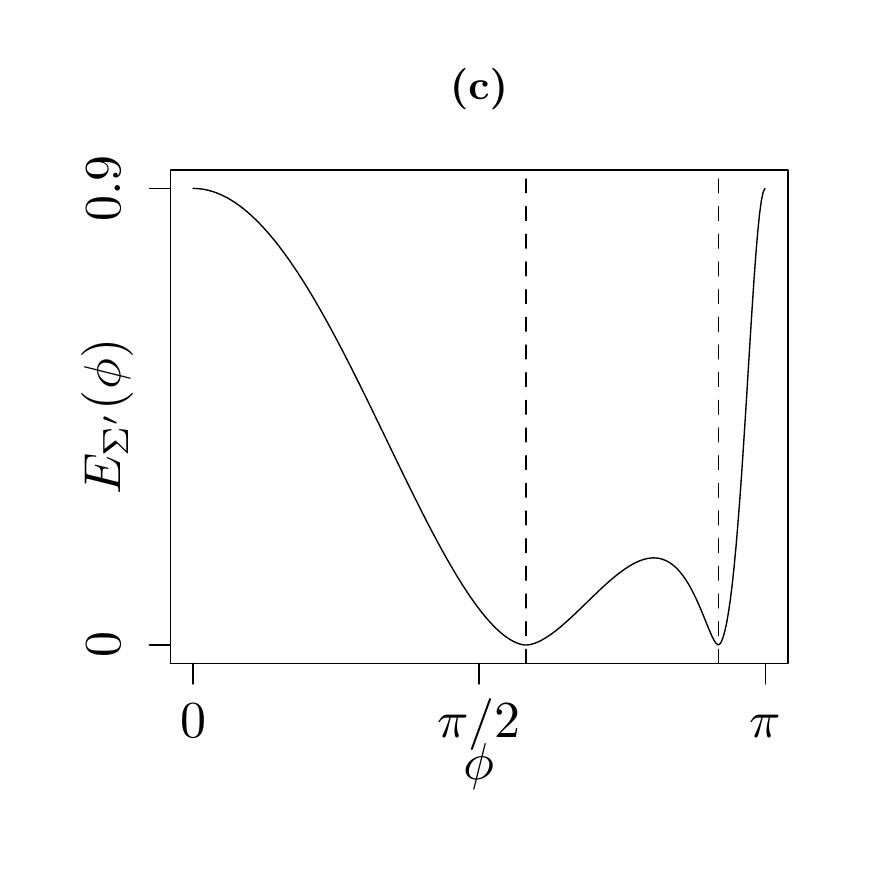}
    \end{subfigure}
    \end{center}
\caption{Fig.(a) illustrates the entanglement in the moving frame of reference as a function of the boosting angle for the parameter values $\eta= 0.6$, $u=v= 0.95$ in the mid-relativistic regime. Fig.(b) shows the Wigner rotation angle as a function of the boosting angle for the speeds $u=v= 0.995$. The dashed vertical lines indicate the interval of the boosting angle, in the ultra-relativistic regime where $\delta \geq \phi/2$. Fig.(c) shows the amount of entanglement in the boosted frame as a function of the boosting angle.}
\label{Fig:local_max}
\end{figure*}

This situation is illustrated in Fig.\ref{Fig:local_max}(a) for the state $\ket{\psi}_\Sigma$.\footnote{The situation for the state $\ket{\xi}_\Sigma$ is similar, the only difference is that monotonic behaviour is inverted.}
The ultra-relativistic regime, where parameter values above the threshold value of the Wigner rotation angle can be observed, is illustrated in Fig.\ref{Fig:local_max}(b) and (c).
We observe, that due to the fact that the entanglement is increasing for $\delta \geq \pi/2$, a local maximum emerges when the entanglement is considered as a function of the boosting angle.
Naturally the question arises whether under specific parameter values the boosted frame of reference observes a higher amount of entanglement than the observer at rest, i.e. if the peak in the entanglement  becomes exceptionally large. The next proposition shows, that the difference in the entanglement in the rest frame and the boosted frame of reference, respectively, is bounded.
\begin{proposition}
Consider the state $\ket{\psi}_\Sigma$. The difference in the entanglement as measured from a moving observer and an observer in the rest frame $\Delta E^\psi  = E^\psi_{\Sigma} - E^\psi_{\Sigma^\prime}$ is bounded by
\begin{align} \label{eq:bound}
\Delta E^\psi \geq \frac{1}{2 \ln 2} \sin^2(2 \eta) \sin^2(\delta).
\end{align}
For the state $\ket{\xi}_\Sigma$, the difference in the entanglement as measured from an observer at rest and a moving observer $\Delta E^\xi  = E^\xi_{\Sigma^\prime} - E^\xi_{\Sigma}$ is bounded by
\begin{align} \label{eq:bounda}
\Delta E^\xi \geq \frac{1}{2 \ln 2} \sin^2(2 \eta) \sin^2(\delta).
\end{align}
\end{proposition}
\begin{proof}
We prove the case $\ket{\psi}_\Sigma$, the results for the state $\ket{\xi}_\Sigma$ can be obtained in a similar way. An observer in the boosted frame of reference observes the same amount of entanglement as the observer in the rest frame in case the Wigner rotation vanishes, i.e. $E^\psi_\Sigma = E^\psi_{\Sigma^\prime}(0)$ holds true. The entanglement difference associated with an arbitrary Wigner rotation angle $\delta$ can therefore be written as
\begin{eqnarray}
\Delta E^\psi  &&=  E^\psi_{\Sigma} - E^\psi_{\Sigma^\prime}=  E^\psi_{\Sigma^\prime}(0) - E^\psi_{\Sigma^\prime}(\delta) \nonumber \\
&&= - \int_0^\delta \frac{\mathrm{d}E_{\Sigma^\prime}}{\mathrm{d} \tilde{\delta}} \mathrm{d} \tilde{\delta}
\end{eqnarray}
By defining the auxiliary variable
\begin{align}
\xi = \xi(\delta) = \sqrt{\cos^2 (2\eta) + \sin^2 \delta \sin^2 (2\eta)},
\end{align}
i.e. $p(\delta) = 1/2(1+ \xi(\delta))$, the derivative of the entanglement in the boosted frame with respect to the Wigner rotation angle stated in \eqref{eq:Wigner_der} is given by
\begin{eqnarray}
\frac{\mathrm{d} E^\psi_{\Sigma^\prime}}{\mathrm{d} \delta} &&= \frac{1}{\ln 2} p^\prime(\delta)[ \ln (1-p(\delta))-\ln p(\delta)] \nonumber \\
&&= \frac{1}{\ln 2} p^\prime(\delta)[\ln(1/2(1-\xi) - \ln(1/2(1+\xi)] \nonumber \\
&&= \frac{1}{\ln 2} p^\prime(\delta)[\ln(1-\xi) - \ln(1+\xi)].
\end{eqnarray}
By reformulating this equation in terms of the inverse hyperbolic tangent via the natural logarithm  $\mathrm{artanh}(x) = 1/2 \ln((1+x)/(1-x))$, we have that
\begin{align}
\frac{\mathrm{d} E^\psi_{\Sigma^\prime}}{\mathrm{d} \delta} &= -\frac{2}{\ln 2} p^\prime(\delta) \mathrm{artanh}(\xi(\delta)).
\end{align}
Thus, we find with $p^\prime(\delta) = \sin 2 \delta \sin^2 2\eta/ (4\xi) $
\begin{align}
\Delta E^\psi  =   \int_0^\delta  \mathrm{artanh} \xi \frac{\sin 2 \tilde{\delta} \sin^2 2\eta}{\xi 2 \ln 2}  \mathrm{d} \tilde{\delta}.
\end{align}
Since $\mathrm{artanh} (x) \geq x$, we find
\begin{eqnarray}
\Delta E^\psi  && \geq \frac{1}{2 \ln 2} \int_0^\delta \sin 2\tilde{\delta} \sin^2 2\eta \mathrm{d} \tilde{\delta} \nonumber\\
&&= \frac{1}{2 \ln 2} \frac{1}{2} \sin^2(2 \eta)(1-\cos (2 \delta)) \nonumber\\
&&=  \frac{1}{2 \ln 2} \sin^2(2 \eta) \sin^2(\delta).
\end{eqnarray}
\end{proof}
We immediately conclude from this proposition, that in case of the state $\ket{\psi}_\Sigma$ the amount of entanglement in the moving frame of reference can never exceed the amount of entanglement measured in the rest frame, since for all parameter values $\Delta E^\psi \geq 0$ holds true.
Therefore, the peak in the entanglement emerging in the ultra-relativistic regime is indeed a local maximum, for non-trivial values of the state-preparation parameter $\eta$ and the Wigner rotation angle $\delta$. In the case of the sate $\ket{\xi}_\Sigma$ the peak is a local minimum due to the fact that the state in the rest frame is initially not entangled.

Note, that we considered states as superpositions of equal and unequal helicity where we assumed the spin-quantization axis to be parallel to the particle's momentum. Relaxing this assumption in order to be able to account for the most general configuration changes the formulas derived, but the qualitative behaviour of the effects in the mid-and ultra-relativistic regime introduced by the boosting angle remains the same.

\section{Conclusion}
In this paper we have analyzed the entanglement of the spin and momentum degrees of freedom for different inertial frames. We explicitly demonstrate that the properties of the amount of the entanglement in the boosted frame as a function of the Wigner rotation angle are different in the mid-relativistic regime ($0 < \delta <  \frac{\pi}{2})$ and the ultra-relativistic regime ($ \frac{\pi}{2} < \delta< \pi$). We show that states with different helicities allow a systematic investigation of the boosting effect. For a state constituted by a linear combination of basis states with equal helicity we demonstrate that in the mid-relativistic regime the entanglement in the boosted frame is minimal, whenever the boosting angle is chosen such that the Wigner rotation is maximal. For $\delta= \frac{\pi}{2}$, characterizing the transition of the mid-into the ultra-relativistic regime we observe that the spin-momentum state in the boosted frame of reference becomes disentangled irrespective of the preparation in the rest frame.  In the ultra-relativistic regime, we show that monotonicity is reversed and that entanglement increases with the Wigner rotation angle. Therefore, we observe an effect that only occurs in the ultra-relativistic regime, i.e. that a local maximum of the entanglement in the boosted frame of reference emerges. The situation is quite different when we consider a state in the rest frame which is constituted by a linear combination of basis states with different helicity. In this case we observe, that monotonic behaviour of the amount of entanglement as a function of the Wigner rotation angle is inverted. Therefore, we observe that in the mid-relativistic regime the entanglement to increase up to the point $\delta= \frac{\pi}{2}$ where we observe maximum entanglement in the boosted frame. In the ultra-relativistic regime $\delta >  \frac{\pi}{2}$ the entanglement increases and a local minimum emerges.

Furthermore, we provided an analytic bound on the difference of the entanglement in the rest frame and the entanglement in the boosted frame.
Therefore, the entanglement in the rest frame can be considered as the ``resource'' of the entanglement in the moving frame of reference and once distributed, it can only (i) decrease in the moving frame of reference for the state $\ket{\psi}$ representing a superposition of equal helicity states in the rest frame and (ii) increase in the moving frame of reference for the state $\ket{\xi}$ representing a superposition of unequal helicity states in the rest frame.

Several publications study entanglement properties in a relativistic regime by considering a specific boosting geometry, where the scenario of perpendicular boosting directions $\phi=  \frac{\pi}{2}$ is the most prominent example. However, for this specific example the ultra-relativistic regime $\delta > \frac{\pi}{2}$ can never be obtained, therefore, disregarding the full parameter space of the boosting angle can be intriguing. 
 Moreover, we have demonstrated by this work that the natural basis for discussing the effects of a change of the reference frame are the that the helicity states underpinning the importance of this property of massive particles and demonstrating the striking difference to photons, where the relativistic nature already forces a distinction for a general non-relativistic consideration.

The counter-intuitive effects of the behaviour of the relative entanglement pinpoints to the difficulty in finding a Lorentz-invariant formulation of a measure of entanglement already for the easy case of pure states in the first quantization. Differently stated, our results show that the geometry of space enters in a non-trivial way effecting the correlation properties, which may pave the path to a novel approach of a Lorentz-invariant entanglement formulation if reachable.

\bibliographystyle{apsrev4-1} 
\bibliography{bibliography} 

\end{document}